\documentclass[a4paper,11pt]{article}
\usepackage[utf8]{inputenc}
\usepackage{url}
\usepackage{amssymb}
\usepackage{amsmath}
\usepackage{amsthm}
\usepackage{graphicx}
\usepackage{caption}
\usepackage{subcaption}
\usepackage[usenames,dvipsnames]{color}
\usepackage{color}
\usepackage{lineno}

\newtheorem{theorem}{Theorem}
\newtheorem{corollary}{Corollary}
\newtheorem{proposition}{Proposition}

\newtheorem{lemma}{Lemma}

\newcommand{\cupdot}{\mathbin{\mathaccent\cdot\cup}}

\newcommand{\G}{\mathcal{G}_{2}(S)}
\newcommand{\ch}{{\rm ch}}

\newcommand{\face}{{\cal F}}

\graphicspath{{images/}}

\date{}

\begin{document}

\title{Geometric Biplane Graphs I: Maximal Graphs
\thanks{A preliminary version of this paper has been presented at the
\emph{Mexican Conference on Discrete Mathematics and Computational Geometry}, Oaxaca, M\'{e}xico, November 2013. 
}}

\author{
Alfredo Garc\'{i}a$^1$\and
Ferran Hurtado$^2$\and
Matias Korman$^{3,4}$\and
In\^{e}s Matos$^5$\and
Maria Saumell$^6$\and
Rodrigo I. Silveira$^{5,2}$\and
Javier Tejel$^1$\and
Csaba D. T\'{o}th$^7$
}

\maketitle

\footnotetext[1]{Departamento de M\'{e}todos Estad\'{\i}sticos, IUMA, Universidad de Zaragoza, Zaragoza, Spain.\\
\url{olaverri@unizar.es, jtejel@unizar.es}.}
\footnotetext[2]{Departament de Matem\`atica Aplicada II, Universitat Polit\`{e}cnica de Catalunya, Barcelona, Spain.\\
\url{ferran.hurtado@upc.edu, rodrigo.silveira@upc.edu}.}
\footnotetext[3]{National Institute of Informatics (NII), Tokyo, Japan. \url{korman@nii.ac.jp}}
\footnotetext[4]{Kawarabayashi Large Graph Project, ERATO, Japan Science and Technology Agency (JST).}
\footnotetext[5]{Departamento de Matem\'atica \& CIDMA, Universidade de Aveiro, Aveiro, Portugal,  \url{ipmatos@ua.pt, rodrigo.silveira@ua.pt}.}
\footnotetext[6]{Department of Mathematics and European Centre of Excellence NTIS (New Technologies for the Information Society), University of West Bohemia, Pilsen, Czech Republic. \url{saumell@kma.zcu.cz}.}
\footnotetext[7]{Department of Mathematics, California State University Northridge, Los Angeles, USA. \url{cdtoth@acm.org}.}

\begin{abstract}
We study biplane graphs drawn on a finite planar point set $S$ in general position. This is the family of geometric graphs whose vertex set is $S$ and can be decomposed into two plane graphs. We show that two maximal biplane graphs---in the sense that no edge can be added while staying biplane---may differ in the number of edges, and we provide an efficient algorithm for adding edges to a biplane graph to make it maximal. We also study extremal properties of maximal biplane graphs such as the maximum number of edges and the largest maximum connectivity over $n$-element point sets.
\end{abstract}


\section{Introduction}
In a \emph{geometric graph} $G=(V,E)$, the vertices are distinct points in the plane in general position, and the edges are straight line segments between pairs of vertices.
A \emph{plane graph} is a geometric graph in which no two edges cross. Every (abstract) graph has a realization as a geometric graph (by simply mapping the vertices into distinct points in the plane, no three of which are collinear), and every planar graph can be realized as a plane graph by F\'ary's theorem ~\cite{F48}. The number of $n$-vertex labeled planar graphs is at least $27.22^n \cdot n!$~\cite{GN09}. However, there are only $2^{O(n)}$ plane graphs on any given set of $n$ points in the plane~\cite{ACN+82,SS11}.

We consider a generalization of plane graphs. A geometric graph $G=(V,E)$ is \emph{$k$-plane} for some $k\in \mathbb{N}$ if it admits a partition of its edges $E=E_1\cupdot\ldots \cupdot E_k$ such that $G_1=(V,E_1), \ldots , G_k=(V,E_k)$ are each plane graphs. Let $S$ be a planar point set in general position, that is, no three points in $S$ are collinear. Denote by $\mathcal{G}_k(S)$ the family of $k$-plane graphs with vertex set $S$. With this terminology, $\mathcal{G}_1(S)$ is the family of plane graphs with vertex set $S$, and $\G$ is the family of 2-plane graphs (also known as \emph{biplane graphs}) with vertex set $S$.

In this and a companion paper~\cite{GHKMSSTT13-II}, we study $\G$ and contrast combinatorial properties between plane graphs $\mathcal{G}_1(S)$ and biplane graphs $\G$. Plane graphs have limitations in achieving some desirable properties, such as high connectivity, as it is known that every plane graph $H$ has a vertex with degree at most $5$, hence $\kappa(H)\le\lambda(H)\le\delta(H)\le5$ (we use standard graph theory notation as in~\cite{ChartrandLesniak}). It is natural to expect that significantly better values can be obtained if the larger family $\G$ is used. This is precisely the topic we explore in these papers, mostly focusing on graph size and vertex connectivity.

One reason for the study of the family $\G$, instead of the most general family $\mathcal{G}_k(S)$, is that testing when a geometric graph is $k$-plane can be done in $O(n\log n)$ time for $k=2$, but it is NP-Complete for any $k\ge 3$~\cite{Epp09}. On the other hand, to imagine biplane graphs we can suppose that the plane has two sides, with the vertices of a graph being on both sides but each edge on only one side. In this way, biplane graphs can model some physical networks, as, for example, printed circuit boards (PCB). (A PCB consists of several electrical components embedded into a board, connected by noncrossing tracks, which can be printed on either side of the board.)

\paragraph{Related concepts.}
Note that the above generalization of plane graphs is reminiscent to, although more restrictive than, the notion of thickness, geometric thickness, and book thickness, which are defined for abstract graphs~\cite{Bei97,Dillencourt}. We recall their definitions for ease of comparison. The \emph{thickness} of an (abstract) graph $G=(V,E)$ is the smallest $k\in \mathbb{N}$ such that $G$ admits an edge partition $E=E_1\cupdot\ldots \cupdot E_k$ with the property that $G_1=(V,E_1), \ldots , G_k=(V,E_k)$ are each planar graphs. The \emph{geometric thickness} of an (abstract) graph $G=(V,E)$ is the smallest $k\in \mathbb{N}$ such that $G$ admits an edge partition $E=E_1\cupdot\ldots \cupdot E_k$ satisfying that $G_1=(V,E_1), \ldots , G_k=(V,E_k)$ can be simultaneously embedded as plane graphs where the vertex set is mapped to a  common labeled point set. The \emph{book thickness} is a restricted version of the geometric thickness where $G_1,\ldots, G_k$ are simultaneously embedded on a point set in convex position.

Notice that every $k$-plane graph, if interpreted as an abstract graph, has geometric thickness at most $k$, but in addition we are given a specific embedding in the plane in which the decomposition into $k$ plane layers is possible. In other words, the term $k$-plane graph refers to a geometric object, a drawing, while having geometric thickness $k$ is a property of the underlying abstract graph. For example, the cycle $C_4$ has geometric thickness $1$, but a drawing connecting the points $(0,0)$, $(1,1)$, $(1,0)$, $(0,1)$ in this cyclic order with straight-line segments would have a crossing and be $2$-plane.

For disambiguation, we also mention two additional notions, which are commonly used in the graph drawing community, but have little to do with our subject. An (abstract) graph is called \emph{$k$-planar} if it has a drawing in the plane (where the edges are Jordan arcs) such that each edge crosses at most $k$ other edges. It is already NP-hard to recognize 1-planar graphs~\cite{KM08}. The other notion worth mentioning is \emph{1-plane}, which is used for a specific geometric drawing of a 1-planar graph~\cite{EHLP12} in which edges are crossed at most once. Note that these two notions have a different meaning from the definition of $k$-plane graphs introduced above.

\paragraph{Prior work and organization of the paper.}
Our main focus is the study of the largest possible graphs for a fixed point set. This involves the concepts of {\em maximum} (graphs with the largest possible number of edges) and \emph{maximal} graphs (i.e., graphs in which the addition of any edge would break the biplane property). In Section~\ref{sec:mini} we formally define both concepts, and study several fundamental properties. Among other results, we show that two maximal biplane graphs on the same point set do not necessarily have the same number of edges. In particular, this implies that the maximum and maximal properties of biplane graphs are not equivalent (as opposed to the case of planar graphs).

Algorithmic issues are studied in Section~\ref{sec_augment}. First, we present an algorithm for determining whether a given geometric graph is biplane or not. We then show how to augment a biplane graph with new edges to a maximal biplane graph. This result is a variant of the fundamental problem of \emph{graph augmentation}, where one would like to add new edges, ideally as few as possible, to a given graph in such a way that some desired property is achieved. There has been extensive work on augmenting a disconnected plane graph to a connected one (see \cite{HT13} for a recent survey) or on achieving good connectivity properties \cite{AGHTU08,ISTW10,AIR09,RW12,Csa12}.

In Section~\ref{sec:extremal}, we study the maximum vertex connectivity that can be attained for the graphs in $\G$ over all $n$-element point sets $S$ in general position. Similar extremal problems have been considered for graphs of thickness or geometric thickness two \cite{Bei97,HSV99}.

In the companion paper \cite{GHKMSSTT13-II} we consider several problems on augmenting plane graphs to biplane supergraphs with higher connectivity, including the case in which the input is only a point set $S$ and the goal is to construct a \emph{good} biplane graph on $S$. These problems are closely related to the results we present here, and have also received substantial attention for the case of plane graphs. 


\section{Fundamental Properties of Maximal Biplane Graphs}
\label{sec:mini}

A (geometric) graph $G=(V,E)$ is \emph{maximal} (or \emph{edge-maximal}) in a family of graphs $\mathcal{F}$ if there is no graph $G'=(V,E')$ in $\mathcal{F}$ such that $E\subset E'$. We are mostly concerned with problems related to finding a biplane graph of high connectivity  or high vertex degree in $\G$ for a given point set $S$. Since the addition of new edges does not decrease the vertex connectivity, we can restrict our attention to maximal graphs in $\G$.

Recall that a maximal plane graph in $\mathcal{G}_1(S)$ is a \emph{triangulation}, that is, a plane graph where all bounded faces are triangles, and the boundary of the outer face is the convex hull $\ch(S)$. It is well known that any two triangulations on the same point set have the same number of edges. Moreover, any triangulation is also  maximal (i.e., the addition of any edge to a triangulation will result in a nonplane graph). We note that these properties do not hold for biplane graphs: there exist point sets $S$ for which not all maximal graphs in $\G$ have the same number of edges (see an example in Figure~\ref{pic:MaximalGraphs}). In particular, this implies that not every maximal biplane graph is maximum.

\begin{figure}[hp]
\centering
\includegraphics[scale=0.6]{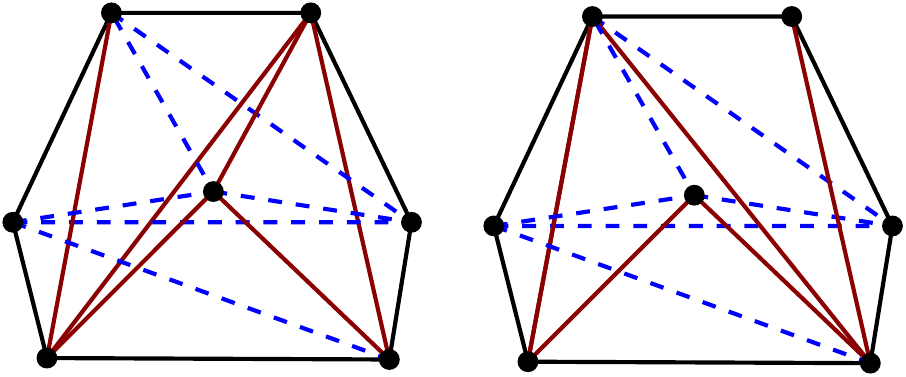}
\caption{Two maximal biplane graphs on the same point set.
The left graph is a maximum biplane graph with 18 edges,
whereas the right one is maximal with 17 edges.}
\label{pic:MaximalGraphs}
\end{figure}

In the following, we show that every maximal biplane graph is the union of two maximal plane graphs.

\begin{lemma}\label{lem:tri}
If $G=(S,E)$ is a maximal biplane graph in $\G$,
then there are two triangulations $T'=(S,E')$ and $T''=(S,E'')$ such
that $E=E'\cup E''$.
\end{lemma}
\begin{proof}
By definition, $G=(S,E)$ has an edge partition $E=E_1\cupdot E_2$, where $G_1=(S,E_1)$ and $G_2=(S,E_2)$ are plane graphs. Augment $G_1$ and $G_2$, independently, to maximal plane graphs $T'=(S,E')$ and $T''=(S,E'')$; thus $T'$ and $T''$ are triangulations. By construction, we have $E\subset E'\cup E''$. The geometric graph $(S,E'\cup E'')$ is biplane by definition. Hence $E=E'\cup E''$, otherwise $G$ would not be maximal.
\end{proof}

The two triangulations, $T'=(S,E')$ and $T''=(S,E'')$, share some edges. The edges of the convex hull $\ch(S)$ are always part of both triangulations, but $T'$ and $T''$ may also share some interior edges. There is a simple characterization of shared edges in terms of edge flips. An edge $e$ in a triangulation is \emph{flippable} if
the union of the two adjacent faces (triangles) is a convex quadrilateral.

\begin{lemma}\label{lem:flip}
Let $G=(S,E)$ be a maximal biplane graph in $\G$ such that $E=E'\cup E''$, where
$T'=(S,E')$ and $T''=(S,E'')$ are two triangulations. No edge of $E'\cap E''$ is flippable in neither $T'$ nor $T''$.
Furthermore, every maximal biplane graph with $n\geq 4$ vertices is 3-connected.
\end{lemma}
\begin{proof}
Suppose, to the contrary, that $e\in E'\cap E''$ is flippable in $T'$ (the case that $e$ is flippable in $T''$ is analogous). We can modify $T'$ by flipping edge $e$. Specifically, let $f$ be the other diagonal of the
convex quadrilateral formed by the two faces of $T'$ adjacent to $e$, and define a new triangulation $T'''=(S,E''')$ with $E'''=(E'\setminus \{e\})\cup \{f\}$. It is clear that $(S,E''\cup E''')$ is biplane, and it contains edge $f$ and all edges in $E$ (including $e\in E''$). Hence it is a biplane graph strictly larger than $G$, contradicting the maximality of $G$. Three-connectivity follows from the fact that a separating chord in a triangulation is always flippable.
\end{proof}

We now study the smallest and largest number of edges that a maximal biplane graph may have.
It is known that every triangulation in $\mathcal{G}_1(S)$ has $3n-h-3$ edges, where $n=|S|\geq 3$ and $h\geq 3$ is the number of vertices of the convex hull $\ch(S)$. By Lemma~\ref{lem:tri}, a maximal biplane graph $G\in \G$ is the union of two triangulations, $T'$ and $T''$, that share the convex hull edges. Thus, it follows that $G$ has at most $6n-3h-6\leq 6n-15$ edges. Hutchinson et al.~\cite{HSV99} improved this bound by showing that, for $n\geq 8$,  every biplane graph in $\G$ has at most $6n-18$ edges. In particular, when $h=3$ the triangulations $T'$ and $T''$ will share at least 3 interior edges. In the remainder of this section we establish lower bounds in terms of $n$ and $h$ for the number of edges in a maximal graph and in
a maximum graph in $\G$.

\begin{theorem}\label{theo_minedges}
Let $S$ be a set of $n\geq 3$ points in the plane such that $\ch(S)$ has $h$ vertices.
Then every maximal graph in $\G$ has at least $\max (\frac{7n}{2}-h-5, 3n-6)$ edges.
Moreover, this bound is tight when $h=n$.
\end{theorem}

\begin{proof}
Let $T_1=(S,E_1)$ be an arbitrary
triangulation of $S$, with $3n-h-3$ edges.  Hoffmann et al.~\cite{hsstw11}
proved that every triangulation of $S$ contains at least $\max (\frac{n}{2}-2, h-3)$
flippable edges.

Now, let $G=(S,E)$ be a maximal biplane graph in $\mathcal{G}_2(S)$, and
suppose that $T_1=(S,E_1)$ and $T_2=(S,E_2)$ are two triangulations such
that $E=E_1\cup E_2$. By Lemma~\ref{lem:flip}, every edge in $E_1\cap E_2$
is flippable in neither $T_1$ nor $T_2$. Thus all flippable edges that are in $E_1$
must be in $E_1\setminus E_2$. Hence, $|E_1\setminus E_2|\ge \max (\frac{n}{2}-2,
h-3)$, and so $|E|=|E_2|+|E_1\setminus E_2|\ge 3n-h-3+\max
(\frac{n}{2}-2, h-3)\ge \max (\frac{7n}{2}-h-5, 3n-6)$.

When $h=n$, we show that the lower bound $\max (\frac{7n}{2}-h-5, 3n-6)=3n-6$ is the best possible. Indeed, every biplane graph on a set of $n$ points in convex position is planar as an abstract graph (Lemma 1(i) of \cite{GHKMSSTT13-II}), hence it has at most $3n-6$ edges.
\end{proof}

\begin{theorem}
Let $S$ be a set of $n\geq 3$ points in the plane such that $\ch(S)$ has $h$ vertices.
Then every maximum graph in $\G$ has at least $4n-h-6$ edges if $h\geq 4$ or $n=3$; and at least $4n-h-7$ edges
if $h=3$ and $n>3$. Moreover, these bounds are tight.
\end{theorem}
\begin{proof}
{\bf Lower bounds.} We proceed by induction on $n$. The base case is $n=h$, where the union of two triangulations of a convex $n$-gon gives a biplane graph with $n+2(n-3)=3n-6=4n-h-6$ edges. In this case, the claim follows directly from Theorem~\ref{theo_minedges}.

Suppose now that $n>h$, and that the claim holds for every set of $n-1$ points whose convex hull has $h$ vertices. Let $s\in S$ be a rightmost point in the interior of $\ch(S)$, and let $S'=S\setminus \{s\}$.
Let $G'=(S',E')$ be a maximum biplane graph on $S'$. By induction, $G'$ has at least $4(n-1)-h-6$ edges if $h\geq 4$ or $n-1=3$; and at least $4(n-1)-h-7$ edges if $h=3$ and $n-1>3$. By Lemma~\ref{lem:tri}, $G'$ is the union of two triangulations $T_1'=(S',E_1')$ and $T_2'=(S',E_2')$.

We construct a biplane graph $G=(S,E)$ by augmenting $G'$ with the new vertex $s$ and some incident edges.
If $h=3$ and $n=4$, then $G'$ is a triangle, and $s$ can only be joined to the 3 vertices of $G'$.
Hence $G$ has $6=4n-h-7$ edges, as required.

If $h\geq 4$ or $n\geq 5$, we join $s$ to at least 4 vertices of $G'$. Point $s$ lies in the interior of a triangle $\Delta'$ of $T_1'$, and a triangle $\Delta''$ of $T_2'$. We can augment $T_1'$ and $T_2'$ each with 3 new edges that join $s$ to the corners of $\Delta'$ and $\Delta''$, respectively, to obtain two new triangulations $T''_1$ and $T''_2$.  If  $\Delta'$ and $\Delta''$ together have at least 4 distinct vertices, then the induction step is complete.

It remains to consider the case in which $\Delta'$ and $\Delta''$ together have only 3 distinct vertices, i.e., $\Delta'=\Delta''$. Since $h\geq 4$ or $n\geq 5$, $T_1''$ cannot be the wheel triangulation. Recall that, by construction, we know that there is no interior point that is to the right of $s$. In particular, the rightmost vertex of $\Delta'$ must be a vertex of $\ch(S)$. It follows that  one of the edges of  $\Delta'$ is flippable in $T_1''$ (proof of this claim is given in Property~3 of~\cite{GHKMSSTT13-II}). Thus, by flipping this edge in $T''_1$ we introduce an additional edge which (together with the three edges from $s$) allows us to augment $G'$ as desired.

\begin{figure}[tb]
\centering
\includegraphics[scale=0.6]{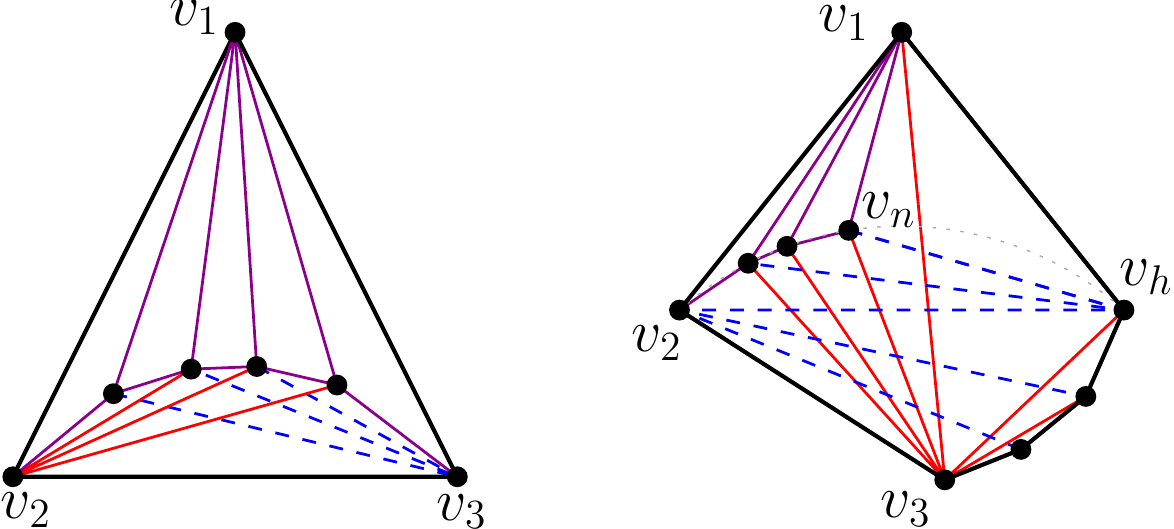}
\caption{Two point sets with convex hulls of size $h=3$ (left) and $h=6$ (right), where
the maximum biplane graph has $4n-h-7$ and $4n-h-6$ edges, respectively.}
\label{fig:MaximumGraphs}
\end{figure}

\noindent{\bf Tightness.} We show that there are point sets that attain the above lower bounds for all $h,n\in \mathbb{N}$, $3\leq h\leq n$.  For $h=3$ and $n>3$, our construction $S$ consists of three vertices of a triangle $\Delta=(v_1,v_2,v_3)$ and $n-3$ points in the interior of $\Delta$, lying on a circular arc between $v_2$ and $v_3$ (Figure~\ref{fig:MaximumGraphs}, left). Every maximal graph in $\G$ contains the edges between $v_1$ and the other $n-1$ vertices, since none of these edges can be crossed by other edges. The $n-1$ points in $S\setminus \{v_1\}$ are in convex position, and admit a biplane graph with $3(n-1)-6$ edges by Theorem~\ref{theo_minedges}. The total number of edges is at most $n-1+(3n-9)=4n-10=4n-h-7$.

For $n=3$, we can simply take a triangle. For $h\geq 4$, our construction of $S$ consists of the vertices of a convex $h$-gon $(v_1,\ldots, v_h)$, and $n-h$ points in the interior of the triangle $\Delta=(v_1,v_2,v_h)$ lying on a circular arc between $v_2$ and $v_h$ such that $v_1v_3$ separates all interior points from the vertices $v_4,\ldots, v_h$ (Figure~\ref{fig:MaximumGraphs}, right). Every maximal graph in $\G$ contains the $h$ hull edges, the $n-h$ edges between $v_1$ and the interior vertices, and the $n-h$ edges of the path $(v_2,v_{h+1},\ldots, v_n)$, since none of these edges can be crossed by other edges. These $2n-h$ edges are part of every triangulation contained in a maximal biplane graph in $\G$.

By Lemma~\ref{lem:tri}, there exist two triangulations $T_R=(S,R)$ and $T_B=(S,B)$ such that $E=R\cup B$. Every triangulation on $n$ points, $h$ of which lie on the hull, has $|R|=|B|=3n-h-3$ edges. As noted above, we have $|R\cap B|\geq 2n-h$. We conclude that $|E|= |R|+|B|-|R\cap B|\leq 2(3n-h-3)-(2n-h)=4n-h-6$.
\end{proof}

\section{Constructing Maximal Biplane Graphs}\label{sec_augment}
We now consider computational aspects related to biplane graphs. The most fundamental algorithmic question is recognition, thus we start by showing how to determine if a given graph is biplane.

\begin{lemma}\label{lem:testing}
Given a geometric graph $G=(S,E)$ with $n$ vertices and $m$ edges,
there is an $O(n\log n)$-time algorithm that tests whether $G$ is biplane
and produces, if possible, a partition $E=E_1\cupdot E_2$ such that
both $(S,E_1)$ and $(S,E_2)$ are plane graphs.
\end{lemma}
\begin{proof}
By the result of Hutchinson et al.~\cite{HSV99}, we know that if  $m>6n-18$ and $n\geq 8$, then $G$ cannot be biplane. The case $n<8$ can be solved by brute force, thus from now on we assume that $m\leq 6n-18$ and $n\geq 8$. Let $G_X$ be the intersection graph of the open line segments in $E$, that is, the nodes of $G_X$ correspond to the edges of $G$, and two nodes are adjacent in $G_X$ if and only if the corresponding edges cross. An edge partition $E=E_1\cupdot E_2$, where $G_1=(S,E_1)$ and $G_2=(S,E_2)$ are plane graphs, corresponds to a bipartition of $G_X$. Given a set of $m$ line segments in the plane, an $O(m\log m)$-time algorithm by Eppstein~\cite{Epp09} returns either an odd cycle in the intersection graph $G_X$ or a 2-coloring of the segments such that segments of the same color are disjoint.\footnote{We note that Eppstein's algorithm does not explicitly construct the intersection graph $G_X$, which may have $\Omega(m^2)$ edges.} Recall that $m\in O(n)$, thus overall the algorithm runs in $O(n\log n)$ time.
\end{proof}

Another natural algorithmic question is how to augment a given biplane graph to a maximal one. That is, given $G\in\G$, can we find a maximal graph $G'\in\G$ such that $G\subseteq G'$? It is easy to augment a plane graph $G\in \mathcal{G}_1(S)$ to a triangulation: we can augment $G$ with all edges of the convex hull $\ch(S)$, and then triangulate each bounded face independently. However, it is not obvious how to augment two layers $E=E_1\cup E_2$ into maximal plane graphs independently. That is, the converse of Lemma~\ref{lem:tri} is not true: if $T'=(S,E')$ and $T''=(S,E'')$ are triangulations, then $G=(S,E'\cup E'')$ is not necessarily maximal biplane.

We can use Lemma~\ref{lem:testing} to greedily augment $G$ into a maximal biplane graph: consider all $O(n^2)$ edges of the complement of $G$ successively, and augment the graph with each new edge as long as the augmented graph remains biplane. Since each test takes $O(n\log n)$ time, we obtain an algorithm that runs in $O(n^3 \log n)$-time. In the following we provide an alternative faster method that exploits the geometric properties of maximal biplane graphs.

\begin{figure}[tb]
        \centering
        \begin{subfigure}[b]{0.45\textwidth}
                \centering
                \includegraphics[page=1,width=0.7\textwidth]{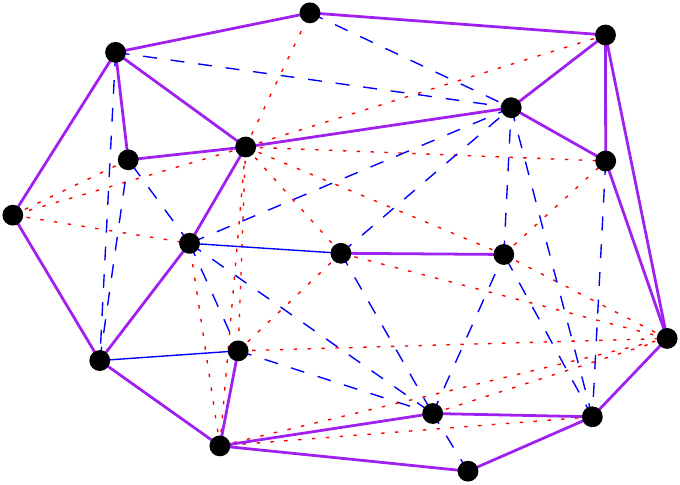}
                \caption{}
                \label{fig:gull}
        \end{subfigure}%
        ~ 
        \begin{subfigure}[b]{0.45\textwidth}
                \centering
                \includegraphics[page=2,width=0.7\textwidth]{example_augmentation}
                \caption{}
                \label{fig:tiger}
        \end{subfigure}
        \caption{(a): A biplane graph with three types of edges: only in red triangulation (dotted), only in blue triangulation (dashed), and in both (solid). (b): Subdivision created by purple edges.}\label{fig:example_augmentation}
\end{figure}

By definition, any graph $G\in\G$ decomposes into two plane graphs. We greedily augment the two plane graphs into
triangulations. Let $T_R=(S,R)$ and $T_B=(S,B)$ be the obtained triangulations, and $G'=(S,R\cup B)$. By construction, we  have $G\subseteq G'$ and $G'\in \G$. Now we classify the edges of $G'$ as \emph{red} if they appear only in $R$, \emph{blue} if they appear only in $B$, or \emph{purple} if they appear in both $R$ and $B$ (see Figure~\ref{fig:example_augmentation}(a)). Let $P=R\cap B$ denote the set of purple edges.

If a purple edge is flippable in $T_R$ or $T_B$, then we can flip it in one triangulation
and retain it in the other, thereby increasing the total number of edges by one (and decreasing the
number of purple edges by one). Intuitively, our algorithm aims to minimize the number of purple edges (thus having the maximum number of edges overall).

A natural approach would be to flip purple edges whenever they are flippable in either $T_R$ or $T_B$. However, as mentioned before, the decomposition $E=B\cup R$ is not unique: it is possible that a purple edge is not flippable in either triangulation, but there is a different decomposition $E=R'\cup B'$ that admits a flippable edge in $R'\cap B'$. To overcome this difficulty, we introduce the concept of a \emph{colorblind flippable} edge.

Consider the plane graph $(S,P)$ formed by all purple edges. The purple graph $(S,P)$ is a subgraph of both triangulations, $T_R$ and $T_B$, and contains all convex hull edges. Each bounded face of $(S,P)$ is a weakly simple polygon (possibly with holes), see Figure~\ref{fig:example_augmentation}(b). Denote by $\face_1,\ldots, \face_k$ the bounded faces of the purple graph $(S,P)$.

Let $e\in P$ be a purple edge that is not an edge of the convex hull. We say that $e$ is \emph{colorblind flippable} (with respect to $R$ and $B$) if it is flippable in the triangulation $T_R$ or $T_B$, or $e$ is adjacent to two different faces of $(S,P)$, and is adjacent to a red triangle in $T_R$ and a blue triangle in $T_B$ forming a convex quadrilateral. With this definition we can obtain a local characterization of maximal biplane graphs. We also use the term \emph{colorblind flippability} to refer to the condition of an edge as being colorblind flippable.

Naturally, the fact that an edge is red or blue depends on the choice of $T_R$ and $T_B$. However, the same does not hold for purple edges. Regardless of the choice of the triangulations, an edge of $G$ is purple if and only if it is not crossed by any other edge of $G$: if it is crossed, then it cannot appear in both triangulations. Otherwise, its insertion cannot break the planarity property of either triangulation (which implies that it was in both triangulations by maximality). Thus, the fact that an edge is purple does not depend on the choice of the triangulations. In fact, we show (Corollary~\ref{cor:purple}) that the colorblind flippability of a purple edge does not depend on the choice of the triangulations, either. More importantly, this observation yields the following characterization of maximal biplane graphs by a local property.

\begin{theorem}\label{theo:colorblind}
Let $G=(S,E)$ be a biplane graph, and let $T_R=(S,R)$ and $T_B=(S,B)$ be two triangulations such that $E=R\cup B$. Then $G$ is a maximal biplane graph if and only if no edge $e\in R\cap B$ is colorblind flippable with respect to $R$ and $B$.
\end{theorem}

Before proving Theorem~\ref{theo:colorblind}, we establish a helpful result (Lemma~\ref{lem:colorchange}). Recall that a bounded face $\face$ of the purple plane graph $(S,P)$ is a weakly simple polygon, possibly with holes. A \emph{chord} of $\face$ is an internal diagonal of $\face$ (connecting two vertices through the interior of $\face$). Denote by $R_{\face}$ (resp., $B_{\face}$) the set of chords of $\face$ in $R$ (resp., in $B$). Note that $R_{\face}$  and $B_{\face}$ define two triangulations of $\face$ which, by definition, satisfy $R_{\face}\cap B_{\face}=\emptyset$. By exchanging the triangulations $R_{\face}$ and $B_{\face}$, we obtain a new decomposition $E=R'\cup B'$ into $R'=(R\setminus R_{\face})\cup B_{\face}$ and $B'=(B\setminus B_{\face})\cup R_{\face}$. We show that all decompositions of $E$ into two triangulations can be obtained by such exchanges in some faces of $(S,P)$.

\begin{lemma}\label{lem:colorchange}
Let $\face$ be a weakly simple polygon possibly with holes, and let $R_{\face}$ and $B_{\face}$ be two disjoint sets of chords, each of which forms a triangulation of the interior of $\face$. Then the intersection graph $G_{\face}$ of the open line segments $R_{\face}\cup B_{\face}$ is connected.
\end{lemma}
\begin{proof}
For an edge $e\in R_{\face}\cup B_{\face}$, denote by $v(e)$ the corresponding node in the intersection graph $G_{\face}$. The chords $R_{\face}$ and $B_{\face}$ form two distinct triangulations of $\face$, which we call {\em red} and {\em blue} triangulations of $\face$. We prove that if $e,e'\in B_{\face}$ are edges of a triangle $\Delta$ in the blue triangulation, then $G_{\face}$ contains a path between the nodes $v(e)$ and $v(e')$. It follows that all blue chords in $B_{\face}$ must be in the same connected component of $G_{\face}$, since the dual graph of the blue triangulation is connected. Analogously, all red chords in $R_{\face}$ are in the same connected component of $G_{\face}$. Since every red edge crosses a blue edge (and vice versa), $G_{\face}$ is connected.

To prove that $G_{\face}$ contains a path between the nodes $v(e)$ and $v(e')$, we consider two situations, see Figure~\ref{fig:aug_connected_component}. If a red chord $e_r\in R_{\face}$ intersects both $e$ and $e'$, then there is a path of length 2 from $v(e)$ to $v(e')$ in $G_{\face}$. Otherwise, there exists a red edge $e_r\in R_{\face}$ that crosses $e$ and a red edge $e_r'\in R_{\face}$ that crosses $e'$ (recall that non-purple edges are crossed by at least one edge of the opposite color). Since $e_r$ and $e_r'$ do not cross each other, neither of them can be incident to any vertex of $\Delta$. Hence, both $e_r$ and $e_r'$ must cross the third edge, $e''$, of $\Delta$. In particular, this implies that $e''\in B_{\face}$, and there are paths of length 2 from $v(e)$ to $v(e'')$ and from $v(e'')$ to $v(e')$. Thus, $v(e)$ and $v(e')$ are connected in $G_{\face}$ by a path of length at most 4.
\end{proof}

\begin{figure}[htb]
\centering
\includegraphics[scale=.9]{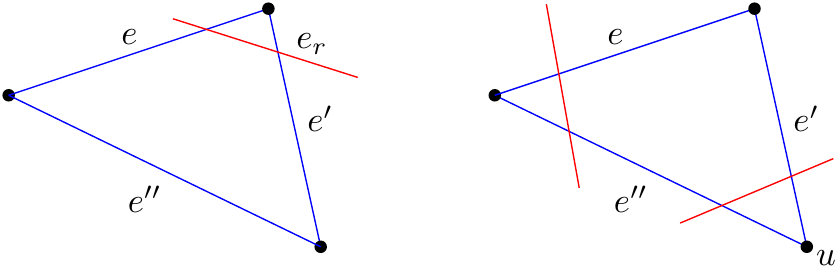}
\caption{ $e,e'\in B_{\face}$ are edges of a triangle $\Delta$ of the blue triangulation.
Left: a red edge $e_r$ crosses both $e$ and $e'$.
Right: no red edge crosses both $e$ and $e'$. Then, there is a red edge that crosses  $e$ and $e''$, and another red edge that crosses $e'$ and $e''$.}
\label{fig:aug_connected_component}
\end{figure}

For each bounded face $\face $ of $(S,P)$, by Lemma~\ref{lem:colorchange}, the intersection graph $G_{\face}$ is bipartite and connected, therefore there is only one way of partitioning the edges of $R_{\face}\cup B_{\face}$ into two plane graphs.
Consequently, if any edge in $R_{\face}$ changes its color, then all edges in $R_{\face}\cup B_{\face}$  must also change their colors. We are now ready to prove Theorem~\ref{theo:colorblind}.

\begin{proof}[Proof of Theorem~\ref{theo:colorblind}.]
Let $G=(S,E)$ be a biplane graph and let $E=R\cup B$ such that $(S,R)$ and $(S,B)$
are two triangulations of $S$. Denote the faces of the purple graph $(S,P)$ by $\face_1,\ldots , \face_k$.

Suppose that there is an edge $e\in E$ colorblind flippable with respect to $R$ and $B$. We show that $G$ cannot be a maximal biplane graph. Clearly, if $e$ is flippable in $R$ or $B$, then we can flip it in one triangulation and retain it in the other, thereby increasing the total number of edges by one. Otherwise, let $\face_i$ and $\face_j$ be the faces of the purple graph $(S,P)$ adjacent to $e$, and suppose without loss of generality that $e$ is adjacent to a red triangle in $R_{\face_i}$ and a blue triangle in $B_{\face_j}$ that form a convex quadrilateral. Then we can obtain a new decomposition $E=R'\cup B'$ with $R'=(R\setminus R_{\face_i})\cup B_{\face_i}$ and $B'=(B\setminus B_{\face_i})\cup R_{\face_i}$. By flipping $e$ in the triangulation $R'$, and retaining it in $B'$, the total number of edges increases by one.

Suppose now that $G=(S,E)$ is not a maximal biplane graph. We show that there is a colorblind flippable edge in $E$ (with respect to $R$ and $B$).
Recall that $R=R_{\face _1}\cup \ldots \cup R_{\face _k}\cup P$ and $B=B_{\face _1}\cup \ldots \cup B_{\face _k}\cup P$, where $P$ is the set of purple edges. Since $G$ is not maximal, $G$ can be augmented to a larger biplane graph $G_{\max}=(S,E_{\max})$ such that $E\subset E_{\max}$. Let $(S,R_{\max})$, $(S,B_{\max})$ be two triangulations such that $E_{\max}= R_{\max} \cup B_{\max}$. Note that, even though $E_{\max}$ contains $E$, it is possible that $R\not\subseteq R_{\max}$ and $B\not\subseteq B_{\max}$.

Since there is only one way of partitioning $R_{\face_i}\cup B_{\face _i}$ into two plane graphs, each $R_{\face_i}$ and $B_{\face _i}$ must be completely contained in either $R_{\max}$ or $B_{\max}$. That is, we have $R_{\max}=C_1\cup\ldots \cup C_k\cup R_{k+1}$ and $B_{\max}=\overline{C}_1\cup\ldots\cup\overline{C}_k\cup B_{k+1}$, where $\{C_i,\overline{C}_i\}= \{R_{\face _i},B_{\face _i}\}$, and $R_{k+1}$ and $B_{k+1}$ are two additional sets that complete $R_{\max}$ and $B_{\max}$ to a triangulation, respectively.

Consider now the triangulations $R'=C_1\cup \ldots\cup C_k\cup P$, and $B'=\overline{C}_1\cup \ldots\cup \overline{C}_k\cup P$. By construction, we have $E=R'\cup B'$, hence this is another decomposition of $E$ into two triangulations. Note that $R'$ and $R_{\max}$ share all edges in $C_1\cup\ldots\cup C_k$.

Consider all triangulations on $S$ that contain the edges $C_1\cup\ldots\cup C_k$ (i.e., these edges are \emph{constrained}), including $R'$ and $R_{\max}$. It is known~\cite{LL86} that between any two constrained triangulations on the same point set, there is a sequence of edge flips (of unconstrained edges) that transform one into the other. In particular, there is a sequence of flips that transforms $R'$ into $R_{\max}$, flipping only \emph{unconstrained} edges. The first edge flipped in the sequence is in $P$, implying that $P$ contains at least one flippable edge with respect to $R'$. With respect to the original decomposition $E=R\cup B$, this edge is colorblind flippable with respect to $R$ and $B$, as required.
\end{proof}
\begin{corollary}\label{cor:purple}
Let $G=(S,E)$ be a biplane graph such that $E$ can be decomposed into two triangulations. Let $(R,B)$ and $(R',B')$ be any two such decompositions (i.e., $E=R\cup B=R'\cup B'$). A purple edge is colorblind flippable with respect to $R$ and $B$ if and only if it is colorblind flippable with respect to $R'$ and $B'$.
\end{corollary}

Before presenting our algorithm, we show that when we flip one edge we cannot alter the colorblind flippability of any other edges that are  ``far'' away.
\begin{lemma}\label{lem_flips}
The flip of an edge $e$ can only change the colorblind flippability of edges in the triangles that contain $e$.
\end{lemma}
\begin{proof}
The fact that an edge $e'$ is colorblind flippable only depends on the two triangles that are adjacent to $e'$ in $R$ and $B$ (more precisely, on the up to four different combinations of adjacent triangles). Thus, flipping $e$ cannot affect the colorblind flippability of $e'$ if $e'$ is not part of one of the four triangles containing $e$. Moreover, by flipping $e$ we reduce the number of purple edges by one. 
\end{proof}

We now describe an algorithm to augment a given biplane graph $G=(S,E)$ to a maximal biplane graph.

Algorithm MAXIMAL:

\begin{enumerate}
\item Compute (using Lemma 3) a decomposition $E=R\cup B$ such that $(S,R)$ and $(S,B)$ are plane graphs.
\vskip 0.3cm

\item Augment $(S,R)$ to a red triangulation $T_R$, and $(S,B)$ to a blue triangulation $T_B$.
\vskip 0.3cm

\item Find the purple edges $P=T_B\cap T_R$, and compute the faces of the purple plane graph $(S,P)$.
For every purple face, compute the set of red (resp., blue) diagonals and create a standard union-find data structure for these sets of diagonals.

\vskip 0.3cm

\item Put all purple edges in a priority queue $Q$.
\vskip 0.3cm

\item For all $e\in Q$: if $e$ is not colorblind flippable, then do nothing, otherwise insert a flipped counterpart of $e$ as a new edge and update the purple face decomposition as follows. If $e$ is flippable in $T_R$ or $T_B$, then the new edge is part of one triangulation, $e$ remains in the other one, and all other edges keep their original color; if e is flippable in the union of a red and a blue triangle that lie in two different purple faces, then all chords in one of the faces adjacent to $e$ in the purple graph change their color. After this recoloration, $e$ becomes flippable in one of the two triangulations, and the algorithm proceeds as in the previous case. In all cases remove $e$ from $Q$.
\vskip 0.3cm

\item After each flip, reinsert into $Q$ the purple edges that now become colorblind flippable. (By Lemma~\ref{lem_flips}, up to four other purple
edges can be affected by the flip).
\vskip 0.3cm

\item The algorithm ends when $Q$ is empty.
\vskip 0.3cm

\end{enumerate}

\begin{theorem}\label{lem:max}
Given a biplane graph $(S,E) \in \G$ the above algorithm computes a maximal graph $(S,E_{\max}) \in \G$ such that $E\subseteq E_{\max}$ in $O(n \log n)$ time.
\end{theorem}
\begin{proof}
The algorithm terminates because after each flip the number of purple edges decreases by one. Since it terminates with a graph containing no colorblind flippable edges, by Theorem~\ref{theo:colorblind}, the obtained graph is maximal.
Thus, it remains to show that the algorithm runs in $O(n\log n)$ time. By Lemma~\ref{lem:testing}, it takes $O(n \log n)$ time to produce the initial decomposition $E=R\cup B$. We can complete both layers into triangulations in $O(n\log n)$ time. We assume the triangulations are represented in a data structure allowing constant-time navigation between edges and adjacent faces (such as a doubly connected edge list). Classification of the edges into red, blue, and purple, as well as creating the face-decomposition of the purple graph, can be done in $O(n\log n)$ time. For each purple edge, we store its two adjacent red triangles and two adjacent blue triangles. Hence, we can check whether a purple edge $e\in P$ is colorblind flippable (with respect to $R$ and $B$) in constant time.

The second phase consists of checking all purple edges and trying to flip them. The algorithm maintains all purple edges in a priority queue. Note that when a purple edge is flipped, its two adjacent faces in the purple graph merge. We maintain the set of red and blue chords of the purple graph in a standard union-find data structure, so we can find which face a chord belongs to and merge two faces in $O(\log n)$ amortized time. In this way, processing each purple edge takes $O(\log n)$ amortized time. In addition to updating the face structure, we must check the flippability of up to four more purple edges each time an edge is added into $G$. We charge this extra cost to the added edge.

Since the number of edges in a biplane graph is bounded by $6n-18$, by Lemma~\ref{lem:tri}, the number of edges we will check is also bounded by $O(n)$. That is, after an $O(n\log n)$-time preprocessing, our algorithm will check the flippability of $O(n)$ edges. Each purple edge can be processed in $O(\log n)$ amortized time. Thus, we conclude that the algorithm runs in $O(n\log n)$ total time.
\end{proof}

\paragraph{Remark.}
The algorithm in Theorem~\ref{lem:max} augments a biplane graph $(S,E)$ drawn on a point set $S$ into a maximal biplane graph adding edges one-by-one. If we need an arbitrary maximal biplane graph on $S$, then we can start with the empty graph $(S,\emptyset)$; if we would like to generate another maximal biplane graph, it suffices to execute the algorithm again where the initial graph consists of a single edge not present in the previously obtained maximal graph. Finally, observe that this procedure will construct maximal graphs, but the resulting graph need not be maximum. Thus, it remains an open problem to efficiently compute a \emph{maximum} biplane graph on a given point set $S$.


\section{Connectivity of Maximal Biplane Graphs}\label{sec:extremal}

In this section we consider the following question. What is the maximum possible connectivity of a graph in $\G$ over all $n$-point sets $S$? In other words, this section studies the problem of finding the value
$$\kappa_{2}(n)=\max_{|S|=n} \max_{G\in\G} \kappa(G).$$

If the points in $S$ are in convex position, then every graph in $\G$ is planar (by Lemma 1(i) of ~\cite{GHKMSSTT13-II}), and thus we cannot construct a 6-connected biplane graph. However, biplane graphs may achieve higher connectivity for certain sets. As already noted, Hutchinson et al.~\cite{HSV99} proved that every biplane graph in $\G$ has at most $6n-18$ edges for $n\geq 8$. Therefore the sum of the vertex degrees is at most $12n-36$, and there is always a vertex of degree at most 11. Consequently, 11 is an upper bound for vertex-connectivity. In the following we show how to construct a biplane graph with minimum vertex degree $10$, and then we modify this construction to obtain an $11$-connected biplane graph. The construction combines elements of a construction by Huntchinson et al.~\cite{HSV99} with fullerene graphs.

Huntchinson et al.~\cite{HSV99} constructed a biplane graph with $n$ vertices and $6n-20$ edges (for sufficiently large values of $n$). The core of their construction is a set of $k^2$ points placed on a $k\times k$ section of the integer grid with coordinates $(i,j)$ for $1\le i\le k$ and $1\le j\le k$. Essentially, a vertex $(i,j)$ is connected to vertices $(i\pm 1, j)$, $(i,j\pm1)$, $(i\pm1,j+1)$, $(i\pm1,j-1)$, $(i+2,j+1)$, $(i-2,j-1)$, $(i+1,j-2)$, and $(i-1,j+2)$ whenever they exist (see Figure~\ref{pic:10Connected}). We say that a vertex $(i,j)$ is a \emph{boundary vertex} if $i\in \{1,k\}$ or $j\in \{1,k\}$, and an \emph{interior vertex} otherwise. The graph is the union of two lattice triangulations, as shown in Figure~\ref{pic:10Connected}. Observe that all interior vertices have degree 10 or higher, but the boundary vertices have lower degree: the degree is 4 at the four corners, 6 at four neighbors of the corners and 7 at all other boundary vertices.

Fullerenes~\cite{BGK12} are planar 3-regular graphs with only hexagonal faces, except for precisely 12 faces that  are pentagonal. It is known that a fullerene of $2k$ vertices can be constructed for $k=10$ or $k\geq 12$. Moreover, there exist fullerenes in which pentagonal faces are sufficiently far from each other: for example, for $k\geq 36$ (and for $k=30$), there exists a fullerene with $2k$ vertices in which there are no two adjacent pentagonal faces~\cite{BGK12}.
 The dual of such a fullerene is a planar graph with $k+2$ vertices and triangular faces, where all vertices have
degree $6$ except for twelve vertices of degree $5$ (which are pairwise nonadjacent if no pentagonal faces are adjacent). Since every fullerene is cyclically 5-edge-connected~\cite{Doslic}, its dual graph is 5-connected~\cite{Bar74}. Thus, by Lemma~1(ii) in~\cite{GHKMSSTT13-II}, the dual graph can be represented as a biplane graph drawn on points in convex position.

We modify the $k\times k$ grid construction of Huntchinson et al.~\cite{HSV99} as follows. Slightly deform the bounding box of the $k \times k$ integer grid in a way that each side becomes a reflex curve. Note that these curves need to be sufficiently flat to maintain the intersection pattern of the nonboundary edges (see Figure~\ref{pic:10Connected}). Attach a $5$-connected biplane graph from the class described above (dual graph of a fullerene in which there are no two adjacent pentagonal faces) to each side of the grid.
Align the 5-connected biplane graphs along each side of the grid such that a vertex of degree 6 in the 5-connected biplane graph is identified with a vertex of degree 6 in the grid (marked with empty dots in Figure~\ref{pic:10Connected}). Denote by $G_{10}(k)$ the resulting graph: it has $k^2$ vertices, and is clearly biplane. Moreover, by construction, all vertices of $G_{10}(k)$ have degree at least $10$.

\begin{figure}[tb]
\centering
  \includegraphics[scale=1]{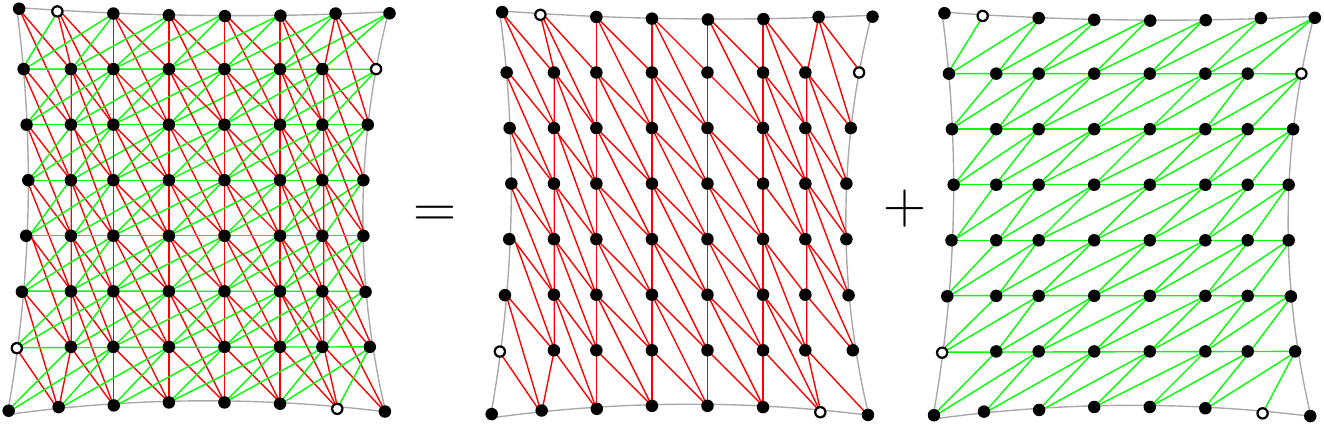}
\caption{The $10$-connected biplane graph formed by two planar triangulations. Empty dots correspond to the vertices of degree 6 in the grid. Each one must also have degree 6 in the 5-connected graph attached to its side.}
\label{pic:10Connected}
\end{figure}

\begin{proposition}\label{prop_10}
For any $k\geq 40$ it holds that $\kappa(G_{10}(k))=10$.
\end{proposition}
\begin{proof}
We refer to the $k^2$ points of $G_{10}(k)=(S,E)$ by their coordinates $(i,j)$ (for $1\le i,j\le k$) in the grid. For a vertex set $U\subset S$, denote by $N(U)$ the set of neighbors of vertices in $U$. Since $G_{10}(k)$ has vertices with degree $10$, we clearly have $\kappa(G_{10}(k))\leq 10$.
We now argue by contradiction that $\kappa(G_{10}(k))> 9$ holds.

Suppose that there is a vertex cut of size at most $9$. Then there is a vertex partition $S=A\cupdot B\cupdot C$ such that $|C|\leq 9$, $A,B\neq \emptyset$, and there is no edge between $A$ and $B$. Note that $N(A)\subseteq A\cup C$ and $N(B)\subseteq B\cup C$. We may assume that $|A|$ is minimal among all such vertex partitions, and by symmetry $|A|\leq n/2=k^2/2$. Since the degree of each vertex is 10 or higher, we have $|A|>1$. We observe the following properties of the set $A\subset S$:

\begin{enumerate}

\item\label{en_noline} Set $A$ cannot contain all the vertices on the vertical line $x=i$ for any $1\leq i\leq k$. Suppose, to the contrary, that $(i,j)\in A$ for all $1\leq j\leq k$. Then for each $j$, either all the vertices of the line $y=j$ are in $A$ or there is a point of $C$ in that line. Since $|C|\le 9$, the latter can happen at most 9 times. Since $k\geq 18$, $A$ contains more than half of the points of $G$ contradicting $|A|<n/2$. By symmetry, $A$ cannot contain all the vertices on the line $y=j$ for any $1\leq j\leq k$.

\item\label{en_9lines} Set $A$ cannot contain vertices in more than 9 rows; otherwise, since $|C|\leq 9$, it should completely contain at least one of these lines. Naturally, the same result is true for columns.

\item\label{en_linei0} By the pigeonhole principle (and the fact that $k\geq 40$), there exists an index $i_0$ such that the line $x=i_0$ contains a vertex of $A$, and either the four previous lines or the four following lines contain no vertices of $A$ (that is, either lines $x=i_0-1,x=i_0-2,x=i_0-3,x=i_0-4$ or lines $x=i_0+1,x=i_0+2,x=i_0+3,x=i_0+4$ are empty of vertices from $A$). Without loss of generality, we assume that the first case holds.
\end{enumerate}

Let $p\in A$ be a vertex adjacent to some $q\in C$. If $q$ is not adjacent to any vertex of $A\setminus \{p\}$, we say that $p$ is a \emph{unique neighbor} of $q$ (in $A$). Consider the case in which a point $q$ has a unique neighbor $p_0\in A$. In this case, we define $A'=A\setminus \{p_0\}$. Notice that $C'=(C\setminus \{q\})\cup \{p_0\}$ is a cut set of size $|C'|= |C|$ that splits $G_{10}(k)$ into subgraphs induced by $A'$ and $B'=B\cup \{q\}$, contradicting the minimality of $A$.

Thus, we conclude that no point in $C$ has a unique neighbor in $A$. Recall that, by property~\ref{en_linei0}, there is no point of $A$ in lines $x=i_0-1$, \ldots , $x=i_0-4$. This implies that there cannot be a point of $A$ in position $(i_0,j)$ for any $j\geq 3$, since otherwise point $(i_0,j)$ would be a unique neighbor of $(i_0-2,j-1)$. So the points of $A$ on the line $x=i_0$ can only be placed in positions $(i_0,1)$ or $(i_0,2)$.

Let $j_0\in \{1,2\}$ be the maximum index such that $(i_0,j_0)\in A$. Then, $A$ must also contain point $(i_0+1,j_0+3)$, otherwise $(i_0,j_0)$ would be the unique neighbor of $(i_0-1,j_0+2)$. By repeating the same argument for point $(i_0+1,j_0+3)$, we see that there cannot be any point in $A$ in position $(i_0+1,j)$ for $j> j_0+3$, otherwise point $(i_0+1,j)$ would be the unique neighbor of $(i_0-1,j-1)$; and also $(i_0+2, j_0+6)\in A$, otherwise $(i_0+1,6)$ would be the unique neighbor of $(i_0,j_0+5)$.  By repeating this argument, we conclude that the points of the form $(i_0+\ell,j_0+3\ell)$ must be in $A$ for all $\ell=0,\ldots , \lfloor k/3\rfloor-1$ (in particular, $i_0\leq \lceil 2k/3\rceil+1$). This contradicts property~\ref{en_9lines}, which states that no more than 9 lines contain points of $A$.
\end{proof}

We emphasize that $G_{10}(k)$ cannot be $11$-connected because it has some 
vertices of degree $10$. Thus, by removing the neighbors of any such vertex, we disconnect $G_{10}(k)$. In the following, we make some local flips around the vertices of degree 10
to increase the minimum vertex degree (and connectivity) to $11$.
For this purpose, we first characterize the vertices that have degree 10 in $G_{10}(k)$. Consider first the vertices along the boundary of the grid and recall that each boundary side of the $k\times k$ grid spans a fullerene. In particular, all but twelve vertices of $G_{10}(k)$ have six neighbors within the same boundary side (regardless of the value of $k$). The remaining twelve vertices will only have 5 neighbors within the boundary side. Note that it is possible to choose where to place these vertices so that, for a sufficiently large $k$, the vertices with 5 neighbors satisfy the following conditions: $(i)$ they are not adjacent in $G_{10}(k)$, $(ii)$ they are sufficiently far apart from any corner (say, with at least 7 grid vertices between one of these vertices and a corner).

The other situation in which a boundary vertex can have degree $10$ is if it has 6 neighbors within the boundary, but only 4 neighbors in the interior of the grid. This only happens to the four vertices at positions $(1,2)$, $(2,k)$, $(k,k-1)$ and $(k-1,1)$. Note that corner vertices have 6 neighbors on each boundary side (and 2 neighbors in the interior of the grid), thus their degree is $14$. That is, regardless of the value of $k\geq 12$, each boundary side will have $12+1=13$ vertices of degree 10.

Finally, we must consider vertices in the interior of the grid that have degree $10$. Note that only those at locations $(2,2)$, $(2,k-1)$, $(k-1,2)$, and $(k-1,k-1)$ will have low degree. It is straightforward to verify that any other vertex in the grid has at least $11$ neighbors (if the vertex is on one of the lines $x=2,x=k-1,y=2,y=k-1$), or $12$ (otherwise). 
Thus, in total we have $4(12+1)+4=56$ vertices of degree $10$. Four of these vertices are located around the grid corners, while the others are spaced along the boundary of the grid.
We increase the degree of these vertices by flipping an edge in an appropriate triangulation.

\begin{figure}[!htb]
\centering
  \includegraphics[width=0.95\textwidth]{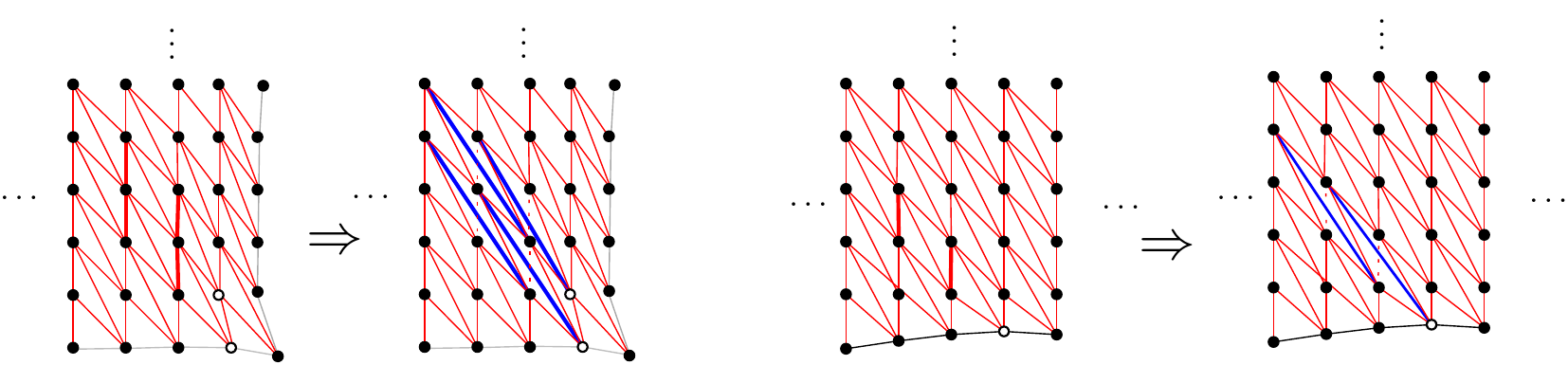}
\caption{To increase the minimum vertex degree around the corner (left) or along the boundary side (right), four and two flips are sufficient, respectively. In both cases, the bold red edges are flipped into the bold blue edges. Empty dots represent the vertices that originally had degree $10$.}\label{pic:11ConnectedA}
\end{figure}

In Figure~\ref{pic:11ConnectedA} (left) we show the changes needed to increase the vertex degree of the vertices at positions $(k-1,1)$ and $(k-1,2)$ (for clarity, only one of the triangulations is shown). Specifically, the four edges $(k-2,2)(k-2,3),(k-2,3)(k-2,4),(k-3,3)(k-3,4),(k-3,4)(k-3,5)$ are replaced by the four edges $(k-1,1)(k-3,4),(k-1,2)(k-3,5),(k-2,2)(k-4,5),(k-2,3)(k-4,6)$. The construction for the other corners is analogous, although the flips might happen in the other triangulation. The case in which the vertex of degree $10$ is along the boundary of the grid can be resolved with only 2 flips (see Figure~\ref{pic:11ConnectedA}, right): the edges $(i-1,2)(i-1,3),(i-2,3)(i-2,4)$ are flipped into the  edges $(i,1)(i-2,4),(i-1,2)(i-3,5)$. Observe that these transformations are local, thus they can be done without affecting each other.

In all, we construct a new graph from $G_{10}(k)$ by replacing $d=2\cdot 56=112$ edges (all of which are parallel to the coordinate axes) with $d$ new edges. Denote by $G'_{10}(k)$ the intermediate graph obtained after deleting these $d$ edges, and by $G_{11}(k)$ the final graph after adding the $d$ new edges. Observe that all vertices have degree $10$ or higher in $G'_{10}(k)$ (and hence in $G_{11}(k)$). For every vertex set $U\subset S$, let $N'(U)$ denote the set of neighbors of $U$ in $G'_{10}(k)$.

By construction the minimum vertex degree is 11 in $G_{11}(S)$, so at least 11 vertices must be  deleted to separate a single vertex from the rest of the graph. We now show that, for separating any larger vertex set $A\subset S$, $2\leq |A|\leq |S|/2$, at least 11 vertices must be deleted in the subgraph $G'_{10}(k)$ of $G_{11}(S)$, if $k$ is sufficiently large. We start with the case $|A|=2$.

\begin{proposition}\label{prop:nodiscon2}
For every vertex set $A\subset S$ of size $|A|=2$, we have $|N'(A)\setminus A|\ge 11$.
\end{proposition}
\begin{proof}
Let $A=\{v_1,v_2\}$ with $v_i=(x_i,y_i)$ for $i=1,2$. Observe that if either $v_1$ or $v_2$ has degree at least $12$, the statement holds. Likewise, if $|x_1-x_2|>4$ or $|y_1-y_2|>4$, we have that $N'(v_1)$ and $N'(v_2)$ do not have points from the interior of the grid in common, and thus the claim follows.

Thus, it remains to consider the case when $|x_1-x_2|\le 4$, $|y_1-y_2|\le 4$, and $|N(v_i)|\le 11$ for $i=1,2$. Since the degree of $v_i$ in $G'_{10}(k)$ is at least 10, we only need to prove that $v_2$ has at least two neighbors outside of $N'(v_1)$ other than $v_1$ itself (or the equivalent statement for $v_1$). By rotating the grid if necessary, we can assume $1\le y_1<y_2<k-1$.
Distinguish between the following five cases:

\begin{enumerate}
  \item If $1<x_2<k$, then the three vertices $(x_2-1,y_2+1),(x_2-1,y_2+2),(x_2+1,y_2+1)\in N'(v_2)$. Observe that none of these vertices can be $v_1$ (since we assumed that $y_1<y_2$). Moreover, at most one of them can be in $N'(v_1)$ (since each point of the grid has at most one adjacency with other vertices whose difference in the $y$ coordinates is at least 2). Thus, we conclude that $N'(v_2) \setminus N'(v_1)$ contains at least two vertices.
  \item If $x_2=1$, then the three vertices $(x_2,y_2+1),(x_2+1,y_2+1),(x_2+2,y_2+1)\in N'(v_2)$. Using the same reasoning, we conclude that among these three vertices, at most one can be discarded.
\end{enumerate}

So far, we have considered the cases in which $v_2$ is in the interior of the grid, or at the left boundary. Since the grid is not symmetric, the case in which $v_2$ is at the right boundary cannot be treated in a similar way. Instead, we consider three more cases.

\begin{enumerate}\setcounter{enumi}{2}
  \item If $x_2=k$ and $x_1<k$, then $(x_2-1,y_2+1),(x_2-1,y_2+2)$ are in $N'(v_2)$, and neither of them can be in $N'(v_1)\cup \{v_1\}$.
  \item If $x_2=k$, $x_1=k$, and $y_1>1$,  then $(x_1-1,y_1-1),(x_1-2,y_1-1)$ are in $N'(v_1)\setminus N'(v_2)$.
  \item If $x_2=k$ and $v_1 = (k, 1)$, then $v_1$ is a corner vertex, and $v_1$ alone has 14 neighbors.
\end{enumerate}\end{proof}

According to Proposition~\ref{prop:nodiscon2} the above result says that if we want to disconnect exactly two vertices from $G'_{10}(k)$, we must remove at least $11$ vertices. A similar result holds for larger sets as well.

\begin{proposition}\label{prop:nodiscon3}
For any $k>491$, and set $A\subset S$ of size $2\leq |A|\leq n/2$,
we have $|N'(A)\setminus A|\ge 11$.
\end{proposition}
\begin{proof}
The proof of this claim is analogous to the proof of Proposition~\ref{prop_10}. The key observation is that only a constant number $d$ of edges have been deleted from $G_{10}(S)$, and all of them are parallel to the coordinate axes.

Suppose, to the contrary, that there is a vertex set $A\subset S$ of size $2\leq |A|\leq n/2$
such that $|N'(A)\setminus A|\leq 10$. Let $A$ be a minimal such set. By Proposition~\ref{prop:nodiscon2}, we have $|A|\geq 3$.
Similar to Proposition~\ref{prop_10}, when $k>3(10+d+1)+10+d=491$, we can prove the following statements.
\begin{enumerate}
  \item Set $A$ cannot contain all the vertices on a line $x = i$ (or $y=j$).
  \item Set $A$ cannot contain vertices in more than $10+d$ rows (columns).
  \item There exists an index $i_0$ such that the line $x=i_0$ contains a vertex of $A$, and either the four previous lines or the four following lines contain no vertices of $A$.
\end{enumerate}

We now use a reasoning analogous to the one given in Proposition~\ref{prop_10}: since $A$ is minimal, no vertex of $C$ can have a unique neighbor in $A$. Thus, there cannot be a vertex of $A$ in position $(i_0,j)$ for any $j\geq 3$. Now, if $j_0\in \{1,2\}$ is the maximum index such that $(i_0,j_0)\in A$, then $A$ must also contain vertex $(i_0+1,j_0+3)$. By repeating this argument with the newly obtained points, we conclude that the points of the form $(i_0+\ell,j_0+3\ell)$ must be in $A$, for $\ell=0,\ldots , k/3-1$. This contradicts the fact that no more than $10+d$ lines contain points of $A$.
\end{proof}

By combining Propositions~\ref{prop:nodiscon2} and~\ref{prop:nodiscon3}, the fact that the minimum vertex degree of $G_{11}(k)$ is $11$, and that $G'_{10}(k) \subset G_{11}(k)$, we obtain a graph of the maximum possible connectivity.

\begin{theorem}\label{theor_11}
There exist infinitely many $11$-connected biplane graphs, and no biplane graph is 12-connected.
\end{theorem}
\section*{Acknowledgements}
A. G.,  F. H., M. K., R.I. S. and J. T. were partially supported by ESF EUROCORES programme EuroGIGA, CRP ComPoSe: grant EUI-EURC-2011-4306, and  by project MINECO MTM2012-30951/FEDER. F. H., and R.I. S. were also supported by project Gen. Cat. DGR 2009SGR1040. A. G. and J. T. were also supported by project E58-DGA. M.~K. was supported by the Secretary for Universities and Research of the Ministry of Economy and Knowledge of the Government of Catalonia and the European Union. I.~M. was supported by FEDER funds through COMPETE--Operational Programme Factors of Competitiveness, CIDMA and FCT within project PEst-C/MAT/UI4106/2011 with COMPETE number FCOMP-01-0124-FEDER-022690. M.~S.\ was supported by the
project NEXLIZ - CZ.1.07/2.3.00/30.0038, which is co-financed by
the European Social Fund and the state budget of the Czech
Republic, and by ESF EuroGIGA project ComPoSe as F.R.S.-FNRS - EUROGIGA NR 13604. R.~S. was funded by Portuguese funds through CIDMA (Center for Research and Development in Mathematics and Applications) and FCT (Funda\c{c}{\~a}o para a Ci{\^e}ncia e a Tecnologia), within project PEst-OE/MAT/UI4106/2014, and by FCT grant SFRH/BPD/88455/2012.


\begin{thebibliography}{00}

\bibitem{AGHTU08}
M. Abellanas, A. Garc\'{\i}a, F. Hurtado, J. Tejel, and J. Urrutia,
Augmenting the connectivity of geometric graphs,
\emph{Computational Geometry: Theory and Applications} {\bf 40} (3) (2008), 220--230.

\bibitem{ACN+82} M. Ajtai, V. Chv\'atal, M.M. Newborn, and E. Szemer\'edi,
Crossing-free subgraphs,
\emph{Annals of Discrete Mathematics} \textbf{12} (1982), 9--12.

\bibitem{ISTW10}
M.~Al-Jubeh, G.~Barequet, M.~Ishaque, D.~L.~Souvaine, C.~D.~T\'oth, and A.~Winslow,
Constrained tri-connected planar straight line graphs,
in \emph{Thirty Essays on Geometric Graph Theory (J. Pach, ed.)}, Springer, 2013, pp. 49--70.

\bibitem{AIR09}
M.~Al-Jubeh, M.~Ishaque, K.~R\'edei, D.~L.~Souvaine, C.~D.~T\'oth, and P.~Valtr,
Augmenting the edge connectivity of planar straight line graphs to three,
\emph{Algorithmica} \textbf{61} (4) (2011), 971--999.

\bibitem{Bar74} D. Barnette,
On generating planar graphs,
\emph{Discrete Mathematics} \textbf{7} (1974), 199--208.

\bibitem{Bei97} L.~W. Beineke,
Biplanar Graphs: A Survey,
\emph{Computers \& Mathematics with Applications} \textbf{34} (11) (1997), 1--8.

\bibitem{BGK12} G. Brinkmann, J. Goedgebeur, B. D. McKay,
The generation of fullerenes,
\emph{Journal of Chemical Information and Modeling} \textbf{52} (11) (2012), 2910--2918.

\bibitem{ChartrandLesniak}
G. Chartrand and L. Lesniak,
\emph{Graphs and digraphs},
Chapman and Hall/CRC, 2005.

\bibitem{Dillencourt} M. B. Dillencourt, D. Eppstein and D. S. Hirschberg,
Geometric Thickness of Complete Graphs,
\emph{J. Graph Algorithms \&  Applications} \textbf{4(3)} (2000), 5--17.

\bibitem{Doslic} T. Doslic,
Cyclical edge-connectivity of fullerene graphs and $(k, 6)$-cages,
\emph{Journal of Mathematical Chemistry} \textbf{33} (2003), 103--112.

\bibitem{EHLP12}
P.~Eades, S.-H.~Hong, G.~Liotta, and S.-H.~Poon,
F\'{a}ry's theorem for 1-planar graphs,
in \emph{Proc. 18th COCOON}, LNCS~7434, Springer, 2012, pp. 335--346.

\bibitem{Epp09} D. Eppstein,
Testing bipartiteness of geometric intersection graphs,
\emph{ACM Transactions on Algorithms} \textbf{5} (2) (2009), article 15.

\bibitem{F48} I. F\'ary,
On straight-line representation of planar graphs,
\emph{Acta Scientiarum Mathematicarum} (Szeged) \textbf{11} (1948), 229--233.

\bibitem{GHKMSSTT13-II}
A.~Garc\'{i}a, F.~Hurtado, M.~Korman, I.~Matos, M.~Saumell, R.~Silveira,
  J.~Tejel, and C.~D. T\'{o}th.
\newblock Geometric biplane graphs {II}: Graph augmentation.
\newblock {\em Graphs {\&} Combinatorics}, 31(2):427--452, 2015.
\newblock Special issue of selected papers from the Mexican Conference on
  Discrete Mathematics and Computational Geometry (2013). Also available on {\em arXiv}
  
\bibitem{GN09} O. Gim\'enez and M. Noy,
Asymptotic enumeration and limit laws of planar graphs,
\emph{Journal of the AMS} \textbf{22} (2) (2009), 309--329.  

\bibitem{hsstw11} M.~Hoffmann, A.~Schulz, M.~Sharir, A.~Sheffer, C.~D.~T\'oth, and E.~Welzl,
Counting plane graphs: flippability and its applications,
in \emph{Thirty Essays on Geometric Graph Theory (J. Pach, ed.)}, Springer, 2013, pp. 303--326.

\bibitem{HT13} F.~Hurtado and C.~D.~T\'oth,
Plane geometric graph augmentation: a generic perspective,
in \emph{Thirty Essays on Geometric Graph Theory (J. Pach, ed.)}, Springer, 2013, pp. 327--354.

\bibitem{HSV99} J. P. Hutchinson, T. C. Shermer and A. Vince,
On representations of some thickness-two graphs,
\emph{Computational Geometry: Theory and Applications} {\bf 13} (1999), 161--171.

\bibitem{KM08} V. P. Korzhik and B. Mohar,
Minimal obstructions for 1-immersions and hardness of 1-planarity testing,
in \emph{Proc. 16th Graph Drawing}, vol.~5417 of LNCS, Springer, 2009, pp. 302--312.

\bibitem{LL86} D. T. Lee and A. K. Lin,
Generalized Delaunay triangulation for planar graphs,
\emph{Discrete and Computational Geometry} \textbf{1} (1986), 201--217.

\bibitem{RW12} I. Rutter and A. Wolff,
Augmenting the connectivity of planar and geometric graphs,
\emph{Journal of Graph Algorithms and Applications} \textbf{16} (2) (2012), 599--628.

\bibitem{SS11} M. Sharir and A. Sheffer,
Counting triangulations of planar point sets,
\emph{Electronic Journal of Combinatorics} \textbf{18} (1) (2011).

\bibitem{Csa12}
C.~D. T\'{o}th,
Connectivity augmentation in planar straight line graphs,
{\em European Journal of Combinatorics}, \textbf{33} (3) (2012), 408--425.

\end{thebibliography}
\end{document}